\documentclass[onecolumn,twoside,11pt]{IEEEtran}

\usepackage{amsmath,amstext,amsthm,amssymb,epsf}
\usepackage[pdftex]{graphicx}

\numberwithin{equation}{section} 

 \newtheorem{lemma}{Lemma}[section]
 \newtheorem{theorem}[lemma]{Theorem}

 \newtheorem{claim}[lemma]{Claim}
 \newtheorem{corollary}[lemma]{Corollary}
 
 \newtheorem{definition}[lemma]{Definition}
 \newtheorem{rem}[lemma]{Remark}
\newenvironment{remark}{\begin{rem}}{\hspace*{\fill}$\diamondsuit$\end{rem}}
 \newtheorem{ex}[lemma]{Example}
\newenvironment{example}{\begin{ex}}{\hspace*{\fill}$\diamondsuit$\end{ex}}
\renewcommand{\emptyset}{\varnothing}

\begin{document}

\title{WEB SIMILARITY IN SETS OF SEARCH TERMS USING DATABASE QUERIES}
\author{Andrew R. Cohen* \thanks
{Andrew Cohen is with the Department of Electrical and Computer Engineering,
Drexel University.
Address: A.R. Cohen, 3120--40 Market Street,
Suite 313, Philadelphia, PA 19104,
USA. Email: {\tt andrew.r.cohen@drexel.edu}}
and Paul M.B. Vit\'{a}nyi
\thanks{
Paul Vit\'{a}nyi is with the national research center for mathematics
and computer science in the Netherlands (CWI),
and the University of Amsterdam.
Address:
CWI, Science Park 123,
1098XG Amsterdam, The Netherlands.
Email: {\tt Paul.Vitanyi@cwi.nl}.
}}

\maketitle

\begin{abstract}
    
  Normalized web distance (NWD) is a similarity or normalized semantic distance
  based on the World Wide Web 
  or another large electronic database, for instance Wikipedia, 
  and a search engine that returns reliable aggregate page counts. 
  For sets of search terms the NWD 
  gives a common similarity (common semantics)
  on a scale from 0 (identical) to 1 (completely different).
  The NWD approximates the similarity of members of a set
  according to all (upper semi)computable properties. We develop the theory
  and give applications of classifying using Amazon,
  Wikipedia, and the NCBI website from the National Institutes of Health. 
  The last gives new correlations between health hazards.
  A restriction of the NWD to a set of two yields the earlier normalized
  google distance (NGD) but no combination of the NGD's of pairs in a set can extract the
  information the NWD extracts from the set. The NWD enables a new contextual 
(different databases) learning approach 
  based on Kolmogorov complexity theory that incorporates knowledge from these databases.

{ACM classification}

(1) CCS -- Information systems--- World Wide Web ---Web searching and information 
discovery 

(2) CCS--- Information Retrieval

\end{abstract}

\section{Introduction}
\label{sect.intro}
Certain objects are computer files that carry all their 
properties in themselves. For example the scanned handwritten digits in the
MNIST database \cite{LCB}. 
However, there are also objects that are given by name, such as
`red,' `three,' `Einstein,' or `chair.' Such objects acquire their meaning
from the common knowledge of mankind. We can give objects either 
as the object itself or as the name of that object, 
such as the literal text of the work ``Macbeth by Shakespeare''
or the name ``Macbeth by Shakespeare.'' 
We focus on the name case and provide semantics
using the background information of a large data base such as
the World Wide Web or Wikipedia,
and a search engine that produces reliable aggregate page counts. 
The frequencies involved enable us to compute a distance for
each set of names. This is the web information distance of that set or
more properly the web information diameter of that set. 
The normalized form of this distance expresses similarity, that is,
the semantics (properties, features) the names in the set have in common.
Insofar as the distance or diameter
of the set as discovered by this process 
approximates the common semantics of the objects in the set in human society, 
the above distance expresses this common semantics. 
%
%
The term ``name'' is used here synonymously
with ``word''  ``search term'' or ``query.'' 
The normalized distance above is called the normalized web distance (NWD). 
To compute $NWD(X)$ of a set 
$X= \{\mbox{\rm name}_1, \ldots , \mbox{\rm name}_n\}$ we just use the number
of web pages returned on the query ``$\mbox{\rm name}_1 \; \ldots \;  
\mbox{\rm name}_n$,'' the minimum number of web pages returned on the query
for a name in $X$, the maximum number of web pages returned on the query
for a name in $X$, and the total number of web pages capable of being
returned.
A restriction of the NWD to a set of two yields the earlier Normalized Google Distance (NGD) \cite{CV07}
but no combination of the NGD's of pairs in a set can extract the
information the NWD extracts from the set as we shall show.
\subsection{Goal}
Suppose we want to classify a new object in the most appropriate 
one of several classes of objects. The objects in each class have
a certain similarity to one another. For example all the
objects may be red, flowers, and so on. We are talking here of properties
which all the objects in a class share.  
Intuitively the new object should go into the class of which the similarity
changes as little as possible under the insertion. Among those we should choose
the class of maximal similarity. A red flower may go into
the class in which all the objects are red flowers.
To achieve this goal we need to define a measure of similarity between
the objects of a class. This similarity measure is associated with the 
class and to compare different classes it should be relativized. Namely,
if in class $C_1$ all objects are 1\% the same and in class $C_2$ all objects
are 50\% the same while all objects in $C_1$ are 1000 times larger than
all objects in $C_2$, then in absolute terms the objects in $C_1$ are more
the same than the objects in $C_2$. Therefore the measure of similarity of
a class should be relative and expressed by a number between 0 and 1. 
The NWD proposed here is such a measure of similarity. 

\subsection{Semantics}\label{sect.cs}
The NWD is an extension to sets of the Normalized Google
Distance (NGD) \cite{CV07} which computes a distance between two names. 
Since we deal with names it may be appropriate to equate ``similarity''
with {\em relative semantics} for a pair of names and {\em common semantics}
for a set of more than two names. 
For example, the common semantics of $\{$red, green,
blue, yellow$\}$ comprises the notion ``color'' and the
common semantics of $\{$one, two, three, four$\}$ comprises the notion
``number.'' 
A theory of common semantics
of a set of objects as we develop it here is based on 
(and unavoidably biased by) a
background contents
consisting of a database and a search engine. An example 
is the set of pages constituting the world-wide-web and 
a search engine like Google.
In \cite{KL05} (see also the many references to related research)
it is shown that web searches for rare two-word phrases
correlated well with the frequency found in traditional corpora,
as well as with human judgments of whether those phrases were natural.
The common semantics relations between a set of objects is 
distilled here from the web pages
by just using the number of web pages in which the names 
of the objects occur, singly
and jointly (irrespective of location or multiplicity).
Therefore the common semantics is that of a particular database 
(World Wide Web, Wikipedia, Amazon, Pubnet) and an associated search engine. 
Insofar as the effects of a database-search engine 
pair approximates the utterances
of a particular segment of human society we can identify the 
NWD associated with a set of objects with the (normalized) common semantics
of that set in that segment of human society.

\subsection{NWD and NGD}
It is impossible in general to use combinations of NGD's to 
compute the common semantics
of a set of more than two names. This is seen as follows.
The only thing one can do using the NGD is to compute the NGD's
between all pairs of members in the set and take the minimum,
the maximum, the average, or something else.
This means that one uses the relative semantics between all pairs of members
of the set but not the 
semantics that all members of the set have in common.
For example, each pair may have a lot of relative semantics but 
possibly different relative semantics for each pair. That these semantics
are different may not be inferable from the NGD's. The conclusion may be
that the members of the set have a lot in common.  But in actual fact the 
set may have little or no semantics in common at all. 

The common semantics
of all names in the set is accounted for by the NWD. 
Therefore using the NWD may give 
very different results from using the NGD's.
An example using Google counts is given by 
homonyms such as ``grave,'' ``iron,'' and ``shower.''
On 18 September 2019 Google gave ``grave iron shower'' 12.900.000 results
indicating that this triple of words have little in common. 
But ``grave iron'' got 168.000.000 results,
``iron shower'' got 478.000.000 results, and
``grave shower'' got 46.000.000 results indicating that each of these three
word pairs have more in common than the word triple.
We defer further discussion to Section~\ref{sect.compare} when the 
necessary formal tools are in place.


\subsection{Classification}\label{sect.cl}
In classification we use the semantics the 
objects in a class have in common. Up till now this was replaced
by other measures such as distances in Euclidean space. The NWD
of a class expresses directly (possibly an approximation of) the common
semantics of the objects in the class. According to Section~\ref{sect.cs} 
this cannot be achieved by combinations of
the relative semantics between pairs of objects in the class.
Therefore classification using the NGD's alone may be inferior
to using the NWD's which take crucial information into account as is shown by
Theorem~\ref{ex.NGDvsNWD}. It shows also
that any method using NGD's 
also has a much larger computational complexity. 

\subsection{Background}
To develop the theory behind the NWD we 
consider the information in individual
objects. These objects are finite and expressed as finite binary strings. 
The classic notion of 
Kolmogorov complexity \cite{Ko65} is an objective measure 
for the information in 
a {\em single} object, and information distance measures the information 
between a {\em pair} of objects \cite{BGLVZ98}. 
To develop the NWD we use the new notion of
common information between {\em many} objects \cite{Li08,CV13}.

\subsection{Related Work}\label{sect.relwork}
To determine word similarity or word associations
has been topical in  cognitive psychology \cite{LD97},
linguistics, natural language processing, search engine theory,
recommender systems, and computer science. 
One direction is to use word (phrases)
frequencies in text corpora to develop measures for word similarity
or word association, see the surveys in \cite{TC03,TKS02}.
A successful approach is Latent Semantic Analysis
(LSA) \cite{LD97} that appeared in various forms in a great
number of applications. LSA and its relation to the 
NGD approach is discussed in \cite{CV07}.
As with LSA, many
other previous approaches of extracting correlations 
from text documents are based
on text corpora that are many order of magnitudes smaller, and that are
in local storage,  and on
assumptions that are more refined, than what we propose.
Another recently successful approach is \cite{MCCD13} 
which uses the large text corpora available at Google
to compute so-called word-vectors of two types: predicting the context
or deducing the word from the context. This brute-force approach yields
word analogies and other desirable phenomena. For example, the word vector
of ``king'' minus that of ``man'' plus that of ``woman'' gives a 
word vector near that of ``queen.'' However, just as the other methods
mentioned it gives no common semantics of a set of words but
only a distance between two words like the NGD. Counterexamples to 
using the NGD as 
in Theorem~\ref{ex.NGDvsNWD} work here too: large relative semantics
between every pair of words of a set may not imply 
large common semantics of these words. 
One needs a relation between all the objects like the NWD does.
The NWD makes use of Internet queries. The database used is the Internet 
which is the largest database on earth but this database is a public
facility which does not need to be stored.
To use LSA we require large text corpora in local storage and
to compute word vectors we 
require even larger corpora of words in local storage than LSA does.
Similarly, \cite{CS04,BbA05} and the many references cited there,
use the web and Google counts to identify 
lexico-syntactic patterns or other data.
Again, the theory, aim, feature analysis,
and execution are different from ours, and cannot
meaningfully be compared. Essentially, the NWD method below
automatically extracts semantic relations between sets of arbitrary objects
from the web in a manner that is feature-free,
up to the data base and search-engine used, and computationally feasible.

In \cite{Li08} the notion is introduced of the information required to
go from any object in a finite multiset 
(a set where a member can occur more than once) of objects 
to any other object in the
set. 
Let $X$ denote a finite multiset of
$n$ finite binary strings defined by 
$\{x_1, \ldots, x_n\}$, the constituting elements 
ordered length-increasing lexicographic.
We identify the $n$th tring in $\{0,1\}^*$ ordered
lexicographic length-increasing with the $n$th natural
number $0,1,2, \ldots .$
We denote the
natural numbers by ${\cal N}$.
A {\em pairing function} $\langle \cdot, \cdot \rangle:
{\cal N} \times {\cal N} \rightarrow {\cal N}$  uniquely encodes two
natural numbers (or strings) into
a single natural number (or string) by a primitive recursive bijection.
One of the best-known ones is the computationally invertible
Cantor pairing function defined by
$\langle a,b \rangle = \frac{1}{2} (a+b)(a+b+1)+a$.

The {\em information distance} in $X$ is defined
by 
\[
EG_{\max}(X) = \min \{|p|:  U(p, \langle x,n \rangle)=X, \; 
\mbox{\rm for all } x \in X\}.
\]
(see Appendix~\ref{sect.kolmcomp} for the undefined notions like the
universal computer $U$).
For instance, with $X=\{x,y\}$ the quantity $EG_{\max}(X)$ is the least 
number of bits in a program to
transform $x$ to $y$ and $y$ to $x$. In \cite{Vi11} the mathematical
theory is developed further and the difficulty of normalization
is shown. In \cite{CV13} the normalization is given, justified,
and many applications are given of using compression to classify
objects given as computer files, for example related to the MNIST data base 
of hand written digits and to stem cell classification. 

\subsection{Results}
The NWD is a similarity (a common semantics) between all search terms 
in a {\em set}. (We use set rather than multiset as in \cite{CV13} 
since a set seems more
appropriate than multiset in the context of search terms.) 
The NWD can be thought of as a diameter of the set. 
For sets of cardinality two this diameter reduces to a distance
between the two elements of the set.  
The NWD can be used for the classification
of an unseen item into one of several classes
(sets of names or phrases). This is required in
constructing classes of more than two members while 
the NGD's as in \cite{CV07} suffice for classes of two members.

The basic concepts like the web events, web distribution, and web code
are given in Section~\ref{sect.google}. These are similar to
what is used in \cite{CV07} for the NGD. 
The remaining derivation and results are 
of necessity new and different. We determine the length of 
a single shortest binary program to compute from any 
web event of a single member 
in a set to the web event associated with the whole set 
(Theorem~\ref{theo.just}). 
The mentioned length is an absolute information distance associated
with the set. It is incomputable (Lemma~\ref{lem.egmax}).
It can be large while a set
has similar members and small when the set has dissimilar members. 
This depends on the relative size of the difference between members.
Therefore we normalize to express the 
relative information distance which we associate with similarity 
between members of the set.
We approximate the incomputable normalized version with 
the computable NWD (Definition~\ref{def.NWD}).
In Section~\ref{sect.compare} we compare the NWD and the earlier NGD 
with respect to the computational complexity
(expressed in required number of queries) and accuracy.
The NWD method requires less queries compared to the NGD method while the latter usually also yields inferior results.
In Section~\ref{sect.theory} we present  properties of the NWD 
such as the range of the NWD (Lemma~\ref{lem.0}), whether and how 
it changes under adding members
(Lemma~\ref{claim.1}),
and that it does not satisfy the triangle inequality 
and hence is not metric (Lemma~\ref{theo.triangle}). 
Theorem~\ref{theo.ideal} and Corollary~\ref{cor.ideal} show that 
the NWD approximates 
the common similarity of the queries
in a set of search terms (that is, a common
semantics). 
We subsequently apply the NWD to various data sets
based on search results from Amazon, Wikipedia and the National Center 
for Biotechnology Information (NCBI) website from the U.S. National 
Institutes of Health in Section~\ref{sect.appl}. 
For the methodology of the examples we refer to Section~\ref{sect.metho}.
We treat strings and self-delimiting strings in Appendix~\ref{sect.set},
computability notions in Appendix~\ref{sect.comp},
Kolmogorov complexity in Appendix~\ref{sect.kolmcomp}, and metric of 
sets in Appendix~\ref{sect.metric}. The proofs are deferred to
Appendix~\ref{sect.proofs}. 

\section{Web Distribution and Web Code}
\label{sect.google}
We give a derivation that
holds for {\em idealized} search engines that return reliable 
aggregate page counts from
their {\em idealized} data bases. 
For convenience we call this the ``web'' consisting of 
``web pages.''
Subsequently we apply the idealized theory to real problems using real
search engines on real data bases.

\subsection{Web Event}\label{sect.gd}
The set of singleton {\em search terms}
is denoted by ${\cal S}$,
a {\em set of search terms} is 
$X=\{x_1, \ldots, x_n\}$ with
$x_i \in {\cal S}$ for $1 \leq i \leq n$,
and ${\cal X}$ denotes the set of such $X$.
Let the set of web pages indexed (possible of being returned)
by the search engine be $\Omega$. 
\begin{definition}\label{def.webevent}
\rm
We define the {\em web event} $e(X) \subseteq \Omega$ by the set of web pages 
returned by the search engine doing a search for $X$
such that each web page in the set contains occurrences
of all elements from $X$. 

If $x,y \in S$ and $e(x)=e(y)$ then
$x \sim y$ and the equivalence class $[x]=\{y\in S: y\sim x \}$.
Unless otherwise stated, we consider all singleton search terms
that define the same web event as the same term. Hence we deal actually
with equivalence classes $[x]$ rather than $x$. However, for ease of
notation we write $x$ in the sequel and consider this to mean $[x]$.
\end{definition}

If $x \in S$ then the {\em frequency} of $x$ is $f(x)=|e(x)|$; 
if $X=\{x_1, \ldots , x_n\}$, then $e(X)=e(x_1) \bigcap
\cdots \bigcap e(x_n)$ and $f(X)=|e(X)|$. 
The web event $e(X)$ embodies all direct context
in which all elements from $X$ simultaneously occur in these web pages.
Therefore web events capture in the outlined sense
all background knowledge about this combination of search terms 
on the web. 

\subsection{The Web Code}\label{sect.gc}
It is natural
to consider code words for web events. We base those code words on the
probability of the event. 
Define the {\em probability} $g(X)$ of $X$ 
as $g(X)=f(X)/N$ with $N= \sum_{X \in {\cal X}} f(X)$.
This probability may change over time,
but let us imagine that the probability 
holds in the sense of an instantaneous snapshot. 
A derived notion is the average number of 
different sets of search terms per web page $\alpha$. Since
$\alpha = \sum_{X \in {\cal X}} f(X) /|\Omega|$ we have $N=\alpha |\Omega|$.

A probability mass function on a known set allows  
us to define the associated prefix-code word 
length (information content) equal to unique decodable code word
length \cite{Kr49,Mc56}. Such a  
prefix code is a code such that no code word is a proper
prefix of any other code word. By the ubiquitous Kraft 
inequality \cite{Kr49},
if $l_1,l_2, \ldots$ is a sequence of positive integers satisfying
\begin{equation}\label{eq.kraft}
\sum_i 2^{-l_i} \leq 1,
\end{equation}
then there is a set of prefix-code words of length $l_1,l_2, \ldots .$
Conversely, if there is a set of prefix-code 
words of length $l_1,l_2, \ldots$
then these lengths satisfy the above displayed equation. By the fact that
the probabilities of a discrete set sum to at most 1,
every web event $e(X)$ having probability $g(X)$ can be 
encoded in a prefix-code word.
\begin{definition}
\rm
The {\em length} $G(X)$ of the {\em web code word} for $X  \in {\cal X}$ is
\begin{equation}\label{eq.g}
G(X)=\log 1/g(X),
\end{equation}
or $\infty$ for $g(X)=0$. The case $|X|=1$ gives the length of the web
code word for singleton search terms.
The logarithms are throughout base 2.
\end{definition}
The web code is a prefix code. The code word associated 
with $X$ and therefore with the web event $e(X)$ can be viewed
as a compressed version of the set of web pages constituting $e(X)$.
That is, the search engine 
compresses the set of web pages
that contain all elements from $X$ into a 
code word of length $G(X)$.
(In the following Definition~\ref{def.egmax} we use the notion of 
$U$ and the prefix Kolmogorov complexity $K$ as in
Appendix~\ref{sect.kolmcomp}.) 
\begin{definition}\label{def.egmax}
\rm
Let $p \in \{0,1\}^*$ and $X \in {\cal X} \setminus S$.
The {\em information $EG_{\max}(X)$ to compute event $e(X)$ from
event $e(x)$ for any $x \in X$} is defined by
$EG_{\max}(X)= \min_p\{|p|:  \mbox{\rm for all}
\; x \in X \; \mbox{\rm we have} \; U(e(x),p)=e(X) \}$.
\end{definition}
In this way $EG_{\max}(X)$ corresponds to the length of a 
single shortest self-delimiting program to compute output $e(X)$ 
from an input $e(x)$
for all $x \in X$.

\begin{lemma}\label{lem.egmax}
The function $EG_{\max}$ is upper semicomputable 
but not computable. 
\end{lemma}

\begin{theorem}\label{theo.just}
$EG_{\max}(X) = \max_{x \in X} \{K(e(X)|e(x))\}$ up to an additive 
logarithmic term $O(\log \max_{x \in X} \{K(e(X)|e(x))\})$ which we
ignore in the sequel.
\end{theorem}
%

To obtain the NWD we must normalize $EG_{\max}$.
Let us give some intuition first.
Suppose $X,Y \in {\cal X}$ with $|X|,|Y| \geq 2$.
If the web events $e(x)$'s are more or less the same for all $x \in X$
then we consider the members of $X$ very similar to each other. 
If the web events $e(y)$'s are very different for different $y \in Y$
then we consider the members of $Y$ to be very different from one another.
Yet for certain such $X$ and $Y$ depending on the cardinalities of $X$ and $Y$
and the cardinalities of the web events of the members of $X$ and $Y$ 
we can have $EG_{\max}(X)=EG_{\max}(Y)$.
That is to say, the similarity is dependent on size.
Therefore, to express similarity of the elements in a set $X$ 
we need to normalize $EG_{\max}(X)$ using the cardinality of $X$ and
the events of its members. Expressing the normalized values
allows us to express the degree in which all elements
of a set are alike. Then we can compare truly different sets.

Use the symmetry of information law \eqref{eq.soi}
to rewrite $EG_{\max}(X)$ 
as $K(e(X)) - \min_{x \in X}\{K(e(x))\}$ up
to a logarithmic additive term which we ignore. 
Since $G(X)$ is computable prefix code 
for $e(X)$, while $K(e(X))$ is the shortest
computable prefix code for $e(X)$, 
it follows that $K(e(X)) \leq G(X)$. 
Similarly $K(e(x)) \leq G(x)$ for $x \in X$.
The search engine $G$ returns frequency $f(X)$ on query $X$ (respectively
frequency $f(x)$ on query $x$).
These frequencies are readily converted into $G(X)$ (respectively $G(x)$)
using \eqref{eq.g}. Replace
$K(e(X))$ by $G(X)$ and $\min_{x \in X}\{K(e(x))\}$ by
$\min_{x \in X}\{G(x)\}$ in $EG_{\max}(X)$. 
Subsequently use as normalizing term $\max_{x \in X}\{G(x)\} (|X|-1)$
which gives the best classification results in Section~\ref{sect.appl}
among several possibilities tried.
This yields the following.
\begin{definition}\label{def.NWD}
\rm
The {\em normalized web distance} (NWD) of $X \in {\cal X}$ 
with $G(X) < \infty$ (equivalently $f(X) > 0)$) is 
\begin{eqnarray}\label{eq.NWD}
NWD(X)& =&\frac{G(X) - \min_{x \in X}\{G(x)\}}
{\max_{x \in X}\{G(x)\}(|X|-1)}
\\&=&  \frac{\max_{x \in X}\{\log f(x)\}- \log f(X)}{
(\log N - \min_{x \in X}\{\log f(x)\})(|X|-1)},
\nonumber
\end{eqnarray}
otherwise $NWD(X)$ is undefined.
\end{definition}
The second equality in \eqref{eq.NWD}, 
expressing the NWD in terms of frequencies,
is seen as follows. We use \eqref{eq.g}.
The numerator is rewritten by 
$G(X)=\log 1/g(X)= \log (N/f(X))=\log N- \log f(X)$
and $\min_{x \in X}\{G(x)\}=
\min_{x \in X}\{ \log 1/g(x)\} = \log N- \max_{x \in X}\{ \log f(x)\}$.   
The denominator is rewritten as $\max_{x \in X}\{G(x)\} (|X|-1)
= \max_{x \in X}\{\log 1/g(x)\}(|X|-1)
= (\log N- \min_{x \in X}\{\log f(x)\})(|X|-1)$. 

\begin{example}\label{ex.1}
\rm
Although Google gives notoriously unreliable counts it serves
well enough for an illustration On our scale of similarity, if
$NWD(X)=0$ then the search terms in the set $X$ are identical, and if
$NWD(X)=1$ then the search terms in $X$ are as different as can be.
In October 2019 searching for ``Shakespeare'' gave   224,000,000 hits;
searching for ``Macbeth'' gave 52,200,000 hits; searching for ``Hamlet''
gave 110,000,000 hits; searching
for ``Shakespeare Macbeth'' gave  26,600,000 hits; searching for
``Shakespeare Hamlet'' gave 38,900,000 hits; and searching  
for ``Shakespeare Macbeth Hamlet'' gave 9,390,000 hits.
The number of web pages which can potentially be
returned by Google was estimated by searching for ``the'' as 
25,270,000,000. Using this number as $N$ we obtain
by \eqref{eq.NWD} the $NWD(\{Shakespeare, Macbeth\})
\approx 0.34$, 
$NWD(\{Shakespeare, Hamlet\}) 
\approx 0.32$ and
$NWD(\{Shakespeare, Macbeth, Hamlet\}) 
\approx 0.26$.
We conclude that Shakespeare and Macbeth have a lot in common,
that Shakespeare and Hamlet have just a bit more in common,
and that taken together the terms Shakespeare, Hamlet, and Macbeth 
are even more similar. The ability to compute the NWD for multiple objects simultaneously,
taking a common measure of shared information across the entire query 
is a unique advantage of the proposed approach.
\end{example}

\begin{remark}
\rm
In Definition~\ref{def.NWD} it is assumed that $f(X)>0$
which, since it has integer values, means $f(X)\geq 1$. 
The case $f(X)=0$ means that there is an $x \in X$
such that $e(x) \bigcap e(X \setminus \{x\})= \emptyset$.
That is, query $x$ is independent of the set of queries $X\setminus \{x\}$,
$x$ has nothing in common with $X \setminus \{x\}$ since
there is no common web page. Hence the NWD is undefined.
The other extreme is that $e(x)=e(y)$ ($x \sim y$) for all $x,y \in X$.
In this case the $NWD(X)=0$.
\end{remark}

\section{Comparing NWD and NGD}\label{sect.compare}
The NGD (see Footnote~\ref{foot.1}) is a distance 
between two names.
The NWD is an extension of the NGD
to sets of names of finite cardinality. It is shown that the
NWD has far less computational complexity than the NGD. Moreover, the
NWD uses information to which the NGD is blind, that is, the common similarity
determined by the NWD is far better than that determined by the NGD.
Possibly each pair of objects has a particular relative semantics (NGD) but not necessarily 
the same relative semantics. Yet if this is always the same quantity of relative semantics we may 
conclude wrongly that the whole set of objects have a single semantics in common.
With the NWD we are certain that it pertains to a single common semantics.

\subsection{Computational Complexity}\label{compcompl}
The number of queries needed for using the NWD is usually much less 
than that using the NGD. 
\footnote{\label{foot.1}Defined in \cite[Eq. (6) in Section 3.4 ]{CV07} as 
\[NGD(x,y)=  \frac{ \max \{ \log f(x), \log f(y) \} - \log f(x,y)}
{ \log N - \min \{ \log f(x), \log f(y) \} }.
\]}
We ignore the cost of the arithmetic 
operations (which is larger anyway in
the NGD case) and of determining $N$
which has to be done in both cases.
There are two tasks we consider.

{\em Computing the common similarity of a set.}
The computational complexity of computing the common similarity
using the NGD with a set of $n$ terms  is as follows.
One has to use the search engine
on the data base to determine the search term frequencies. This requires
$n+{n \choose 2}$ frequency computations, namely the frequencies 
of the singletons and of the pairs.
To computational complexity of computing the common similarity of the
same set of $n$ terms by the NWD requires $n$ queries to determine
the singleton frequencies and 1 query to determine the frequency
of pages containing the entire set, that is, $n+1$ times computing frequencies.
Hence computational complexity using the NGD is much higher for large $n$
than that using the NWD.

{\em Classifying.} 
Let $n$ be the total number of elements divided over classes 
$A_1, \ldots , A_m$ of
cardinalities $n_1, \ldots , n_m$, respectively, with $\sum_{i=1}^m n_i = n$.
We classify a new item $x$ into one of the $m$ classes 
according to which class achieves the minimum common similarity (CS) difference
$CS(A \bigcup \{x\})-CS(A)$. If there are more than one such classes we select a
class of maximal CS. We compute the CS using the NGD or the NWD.
Using the NGD we require $n+\sum_{i=1}^m {n_i \choose 2}$ queries to determine
$CS(A_1), \ldots , CS(A_m)$. (Trivially
$\sum_{i=1}^m {n_i \choose 2} \leq {n \choose 2}$). To determine subsequently 
$CS(A_1 \bigcup \{x\}), \ldots ,
CS(A_m \bigcup \{x\})$ we require 1 query extra to determine $f(x)$
and $n$ queries extra to determine $f(x,y)$ for every item $y$ 
among the original
$n$ elements. Altogether there are $2n+1+\sum_{i=1}^m {n_i \choose 2}$
queries required using the NGD.

Using the NWD requires $\sum_{i=1}^m (n_i+1)=n+m$ queries to 
determine the NWD of $A_1, \ldots , A_m$. To subsequently determine the NWDs of
$A_1 \bigcup \{x\}, \ldots , A_m \bigcup \{x\}$ we extra require $f(x)$
and each of $f(\{y: y \in A_i\} \bigcup \{x\})$ for $1 \leq i \leq m$.
That is, $1+m$ queries. So in total $n+2m+1$ queries. 

To classify many new items we may consider training 
cost and testing cost.  
{\em Training cost} is to pre-compute all the queries required for classifying 
a new element---without the costs for the new element. This is only done once. 
{\em Testing cost} is how many queries are required for 
each new item that comes along.
Above we combined these two in the case of one new element.

The training cost for the NGD is up to $n+{n \choose 2}$. 
The testing cost for each new item is $n+1$. 

The training cost for the NWD is $n+m$. The testing cost for each new item
is $m+1$.

\subsection{Extracted Information}\label{sect.nwdvsngd}
Let $A,B$ be two sets of queries and $B \subset A$. Then the common similarity
of the queries in $A \setminus B$ may or may not agree with the common
similarity of the queries in $B$ but adding $A \setminus B$ to $B$
to obtain $A$ will not increase the common similarity of the queries in $A$
above that in $B$. Therefore the common similarity in $A$ is at most that
in $B$. This is generally followed by the NWD without the normalizing
factor $|X|-1$ in the denominator, see Lemma~\ref{claim.1}, except
in the pathological case when condition \eqref{eq.cond} does no hold.

Assume that $A=\{a_1, \ldots , a_n\}$ and $B=\{b_1, b_2\}$ with 
$b_1,b_2 \in A$. Then $NWD(A) \leq \min_{b_1,b_2 \in A} NWD(B)=
\min_{b_1,b_2 \in A} NGD(b_1,b_2)$. Only in this sense using the NGD to
determine the common similarity in a set $A$ gives an upper bound
on $NWD(A)$. All formulas using only NGD's use a subset of the $f(a_i)$'s
and the $f(a_i,a_j)$'s ($1 \leq i,j \leq n$). 
The NWD uses the $f(a_i)$'s and $f(a_1, \ldots , a_n)$. 
For given $f(a_i)$
and the $f(a_i,a_j)$ ($1 \leq i,j \leq n$) the values of 
$f(a_1, \ldots , a_n)$ can be any value in the 
interval $[0,\min_{b_1,b_2 \in A} NGD(b_1,b_2)]$. Hence 
the NWD can vary a lot (and therefore the common similarity) 
for most fixed values of the NGD's.

\begin{example}\label{ex.NGDvsNWD}
\rm
Firstly, we give an example where the common similarity computed from NGD's is
different from that computed by the NWD.
Let $f(x)=f(y)=f(z)= N^{1/4}$ be the cardinalities of the sets of 
web pages containing 
occurrences of the term $x$, the term $y$, and the term 
$z$, respectively. The quantity $N$ is the total number of web pages multiplied
by the appropriate constant $\alpha$ as in Section~\ref{sect.gc}. Let further, 
$f(x,y)=f(x,z)=f(y,z)=N^{1/8}$ and $f(x,y,z)=N^{1/16}$. Here $f(x,y)$ is the 
number of pages containing both terms $x$ and $y$, and so on. Computing
the NGD's gives $NGD(x,y) = NGD(x,z)=NGD(y,z) =1/6$. 
Using for the set $\{x,y,z\}$
either the minimum NGD, the maximum NGD, or the average NGD, will always
give the value $1/6$.
Using the NWD as in \eqref{eq.NWD} we find $NWD(\{x,y,z\})=1/8$.  
This shows that  in this example the common similarity determined using the NGD
is smaller than the common similarity determined using the NWD. 
(Recall that the common 
similarity is 0 if it is maximal and 1 if it is minimal.)

Secondly, we give an example of a difference in classification between 
the NGD and the NWD. The class is selected where the absolute difference in 
common similarity with and without inserting the new item is minimal.
If  more than one class is selected we choose a class with maximal
common similarity. The frequencies of $x,y,z$ and the pairs $(x,y),(x,z),(y,z)$
are as above. For the terms $u,v$ and the pairs $(u,v),(u,z),(v,z)$
the frequencies are $f(u)=f(v)=N^{1/4}$ and 
$f(u,v)=f(u,z)=f(v,z)=N^{1/9}$. 
Suppose we classify the term $z$ into classes $A=\{x,y\}$ and $B=\{u,v\}$ 
using a computation with the NGD's. Then the class $B$ will be selected. 
Namely, the insertion of $z$ in class $A$ will induce new NGD's 
with all exactly having the values of $1/6$ (as above). 
Since $NGD(u,v) = NGD(u,z)= NGD(v,z)=5/36$
insertion of $z$ into the class $B=\{u,v\}$ will
give the NGD's of all resulting pairs $(u,v), (u,z), (v,z)$  values of $5/36$.
The choice being between classes $A$
and $B$ we see that in neither class the common similarity
according to the NGD's is changed. Therefore we select the class where all
NGD's are least (that is, the most common similarity) 
which is $B=\{u,v\}$. Next we select according to the NWD. 
Assume $f(u,v,z)=N^{1/10}$. Then $NWD(u,v,z)=1/4$. 
Then $NWD(\{u,v,z\})-NWD(\{u,v\})(=NGD(u,v)) = 1/4-5/36 =4/36$.
Since $NWD(\{x,y,z\})-NWD(\{x,y\})(=NGD(x,y))=1/8-1/6=-1/24$ 
and selection according to the NWD chooses the least absolute difference
we select class $A=\{x,y\}$. 
\end{example}

\section{Theory}\label{sect.theory}
Let $X=\{x,y\} \in {\cal X}$. 
The NGD distance between $x$ and $y$ in Footnote~\ref{foot.1} equals
$NWD(X)$ up to a constant. 

{\em Range} First we consider the range of the NWD.
For sets of cardinality greater or equal to two the following holds.
\begin{lemma}\label{lem.0}
Let $X \in {\cal X} \setminus S$ and $N > |X|$.
Then $NWD(X) \in [0, (\log_{|X|} (N/|X|)) /(|X|-1)]$.
\end{lemma}
(In practice the range is from 0 to 1; the higher values are theoretically 
possible but seem not to occur in real situations.)

{\em Change for Supersets} We next determine bounds on how the NWD may change under addition of members to
its argument. These bounds are necessary loose since the added members
may be similar to existing ones or very different.
In Lemma~\ref{claim.1} below we shall distinguish two cases related
to the minimum frequencies.
The second case divides into two subcases depending on whether the 
Equation \ref{eq.cond} below holds or not:
\begin{equation}\label{eq.cond}
\frac{f(y_1)f(X)}{f(x_1)f(Y)} \geq
\left(\frac{f(x_0)}{f(y_0)}\right)^{(|X|-1)NWD(X)},
\end{equation}
where $x_0= \arg \min_{x \in X}\{\log f(x)\}$,
$y_0=  \arg \min_{y \in Y} \{\log f(y)\}$,
$x_1= \arg \max_{x \in X}\{\log f(x)\}$, and
$y_1=  \arg \max_{y \in Y} \{\log f(y)\}$.

\begin{example}
\rm
Let $|X|=5$, $f(x_0)=1,100,000$, $f(y_0)=1,000,000$, 
$f(x_1)=f(y_1)=2,000,000$, $f(X)=500$, $f(Y)=100$, and $NWD(X)=0.5$. 
The righthand side of the inequality 
\eqref{eq.cond} is $1.1^2=1.21$ while the lefthand side is 
$5$. Therefore \eqref{eq.cond} holds.
It is also possible that inequality \eqref{eq.cond} does not hold, that is,
it holds with the $\geq$ sign replaced by the $<$ sign.
We give an example. Let $|X|=5$, 
$f(x_0)=1,100,000$, $f(y_0)=1,000,000$, $f(x_1)=f(y_1)=2,000,000$,
$f(X)=110$, $f(Y)=100$, and $NWD(X)=0.5$.
The righthand side of the inequality
\eqref{eq.cond} with $\geq$ replaced by $<$ is $1.1^2=1.21$ 
while the lefthand side is $1.1$. 
\end{example}
 
\begin{lemma}\label{claim.1}
Let $X,Z \subseteq Y$, $X,Y,Z \in {\cal X} \setminus S$, and 
$\min_{z \in Z} \{f(z)\} = \min_{y \in Y} \{f(y)\}$.

{\rm (i)} If $f(y) \geq \min_{x \in X} \{f(x)\}$ for
all $y \in Y$ then $(|X|-1)NWD(X) \leq (|Y|-1)NWD(Y)$. 
{\rm (ii)} Let $f(y) < \min_{x \in X} \{f(x)\}$ for
some $y \in Y$. If \eqref{eq.cond} holds then 
$(|X|-1)NWD(X) \leq (|Y|-1)NWD(Y)$. If \eqref{eq.cond} does not hold then 
$(|X|-1)NWD(X) > (|Y|-1)NWD(Y) \geq (|Z|-1)NWD(Z)$. 
\end{lemma}

\begin{example}
\rm
Consider the Shakespeare--Macbeth--Hamlet Example~\ref{ex.1}.
Let $X =\{Shakespeare, Macbeth\}$, $Y=\{Shakespeare, Macbeth, Hamlet\}$,
and $Z=\{Shakespeare, Hamlet\}$. Then inequality \eqref{eq.cond} 
for $X$ versus $Y$
gives $(124,000,000 \times 7,730,000/(124,000,000 \times 663,000) 
\geq  (22,400,000/22,400,000)^{0.395}$ (that is
$11.659 \geq  1$), and 
for $Z$ versus $Y$ gives
$18,500,000/663,000 \geq (51,300,000/22,400,000)^{0.306}$ (that is
$27.903  \geq  1.289$). In  the first case Lemma~\ref{claim.1} item (i)
is applicable since the frequency minima of $X$ and $Y$ are the same.
(In this case inequality \eqref{eq.cond} is not needed.)
Therefore $NWD(X)(|X|-1)/(|Y|-1) \leq NWD(Y)$ which works out as
$0.395/2 \leq 0.372$. 
In the second case Lemma~\ref{claim.1} item (ii)
is applicable since the frequency minima of $Z$ and $Y$ are not
the same. Since inequality \eqref{eq.cond} holds this gives 
$NWD(Z)(|Z|-1)/(|Y|-1) \leq NWD(Y)$ which works out as
$0.306/2 \leq 0.372$. 
\end{example}

\begin{remark}
\rm
To interpret Lemma~\ref{claim.1} we give the following intuition.
Under addition of a member to a set there are two opposing
tendencies on the NWD concerned.
First, the range of the NWD decreases by Lemma~\ref{lem.0} and
the definition \eqref{eq.NWD} of the NWD 
shows that addition of a member tends to decrease the value of the NWD, 
that is, it moves closer to 0. 
Second, the common similarity and hence the similarity 
of queries in a given set as measured
by the NWD is based
on the number of properties all members of a set have in common.
By adding a member to the set clearly the number of common properties
does not increase and generally decreases. 
This diminishing tends to cause the NWD to possibly increase---move
closer to the maximum value of the range of the new set (which is smaller
than that of the old set). 
The first effect may become visible
when $(|X|-1)NWD(X)> (|Y|-1) NWD(Y)$, which happens
in the case of Lemma~\ref{claim.1} item (ii) for the case when the frequencies 
do not satisfy \eqref{eq.nwdxy}.
The second effect may become visible when $(|X|-1)NWD(X) \leq (|Y|-1)NWD(Y)$,
which happens in Lemma~\ref{claim.1} item (i), 
and item (ii) with the frequencies
satisfying \eqref{eq.nwdxy}.
\end{remark}

{\em Metricity} For every set $X$ we have that the $NWD(X)$ is invariant under 
permutation of $X$: it is {\em symmetric}. 
The NWD is also {\em positive definite} as 
in Appendix~\ref{sect.metric} (where equal members should be interpreted
as saying that the set has only one member).
However the NWD does {\em not satisfy the triangle inequality} and hence is
not a metric. This is natural for a common similarity or semantics:
The members of a set $XY$ (shorthand for $X \bigcup Y$)
can be less similar (have greater NWD) 
then the similarity
of the members of $XZ$ plus the similarity of the members of $ZY$
for some set $Z$. 
\begin{lemma}\label{theo.triangle}
The NWD violates the triangle inequality.
\end{lemma}

{\em Similarity Explained} 
It remains to formally prove that the NWD expresses in the similarity
of the search terms in the set. 
We define the notion of a distance on these sets using
the web as side-information. For a set $X$ a distance (or diameter) of $X$
is denoted by $d(X)$.
We consider only distances  that are 
upper semicomputable, that is,
the distance can be computably approximated
from above (Appendix~\ref{sect.comp}).
A priori we allow asymmetric distances, but we
exclude degenerate distances such as $d(X)=1/2$
for all $X \in {\cal X}$ containing a fixed element $x$. That is,
for every $d$ we want only finitely many
sets $X \ni x$ such that $d(X) \leq d$. Exactly how fast we want the
number of sets we admit to go to $\infty$ is not important; it is only a
matter of scaling. 
\begin{definition}\label{def.wdf}
\rm
A {\em web distance function} (quantifying the  common properties or
common features)
$d: {\cal X} \rightarrow {\cal R}^+$ is {\em admissible} if $d(X)$ is
(i) a nonnegative total real function and is 0 iff $X \in S$;
(ii) it is upper semicomputable from the $e(x)$'s with $x \in X$
and $e(X)$; and (iii) it satisfies 
the density requirement: for every $x \in S$ 
\[
\sum_{X \ni x, \; |X| \geq 2} 2^{-d(X)} \leq 1.
\] 
\end{definition}
We give the gist of what we are about to prove. 
Let $X=\{x_1,x_2, \ldots, x_n\}$. A feature of a query is a property
of the web event of that query. For example, 
the frequency in the web event of web pages containing an occurrence 
of the word ``red.'' We can compute 
this frequency for each $e(x_i)$ ($1 \leq i \leq n$). 
The minimum of those frequencies is the maximum of the number
of web pages containing the word ``red'' which surely
is contained in each web event $e(x_1), \ldots , e(x_n)$.
One can identify this maximum with the inverse of a distance in $X$.
There are many such distances in $X$.
The shorter a web distance is, the more dominant is the feature
it represents.
We show that the minimum admissible distance is $EG_{\max}(X)$.
It is the least admissible web distance and 
represents the shortest of all admissible web distances
in members of $X$. Hence the closer the numerator
of $NWD(X)$ is to $EG_{\max}(X)$ the better it represents the 
dominant feature all members of $X$ have in common.  

\begin{theorem}\label{theo.ideal}
Let $X \in {\cal X}$. The function $G(X)-\min_{x \in X}\{G(x)\}$ 
is a computable upper bound on 
$EG_{\max}(X)$. The closer it is to 
$EG_{\max}(X)$, the better it approximates the shortest admissible
distance in $X$. The normalized form of $EG_{\max}(X)$ is $NWD(X)$. 
\end{theorem}
The normalized least admissible distance in a set is the 
least admissible distance between its members which we call the common
admissible similarity. Therefore we have:
\begin{corollary}\label{cor.ideal}
\rm
The function $NWD(X)$ is the
common admissible similarity among all search terms in $X$.
This admissible similarity can be viewed as
semantics that all search terms in $X$ have in common. 
\end{corollary}

\section{Applications}\label{sect.appl}
\subsection{Methodology}\label{sect.metho}
The approach presented here requires the ability to query a database for the number of occurrences and co-occurrences of the elements in the set that we wish to analyze. One challenge is to find a database that has sufficient breadth to contain a meaningful numbers of co-occurrences for related terms. As discussed previously, an example of one such database is the World Wide Web, with the page counts returned by Google search queries used as an estimate of co-occurrence frequency. There are two issues with using Google search page counts. The first issue is that Google limits the number of programmatic searches   in a single day to a maximum of 100 queries, and charges for queries in excess of 100 at a rate of up to \$50 per thousand. 
The second issue with using Google web search page counts is that the numbers are not exact, but are generated using an approximate algorithm that Google has not disclosed. For the questions considered previously \cite{CV07} we found that these approximate measures were sufficient at that time to generate useful answers, especially in the absence of any a priori domain knowledge. It is possible to implement internet based searches without using search engine API's, and therefore not subject to daily limit. This can be accomplished by parsing the HTML returned by the search engine directly. The issue with google page counts in this study being approximate counts based on a non-public algorithm was more concerning as changes in the approximation algorithm can influence page count results in a way that may not reflect true changes to the underlying distributions.  Since any internet search that returns a results count can be used in computing the NWD, we adopt the approach of using web sites that return exact rather than approximate page counts for a given query.

Here we describe a comparison of the NWD using the set formulation based on web-site search result counts with the pairwise NWD formulation. The examples  are based on search results from Amazon, Wikipedia and the National Center for Biotechnology Information (NCBI) website from the U.S. National Institutes of Health. The NCBI website exposes all of the NIH databases searchable from  a single web portal. We consider example classification questions that involve partitioning a set of words into underlying categories. For the NCBI applications we compare various diseases using the loci identified by large genome wide association studies (GWAS). For the NWD set classification, we  determine whether to assign element $x$ to class $A$ or class $B$ (both classes pre-existing) by computing 
$NWD(Ax)-NWD(A)$ and $NWD(Bx)-NWD(B)$ and assigning element $x$ to 
whichever class achieves the minimum difference. 
A combination of pairwise NGD's for each class suffers in many cases 
from shortcomings as pointed out before and formally 
in Example~\ref{ex.NGDvsNWD}. Therefore, with the aim of doing better,
for the pairwise NWD we use an approach based on spectral clustering. 
Rather than using a combination of simple pairwise information distances 
(NGD's), the spectral approach constructs a representation of the 
objects being clustered using an eigen-decomposition. 
In previous work we have found such spectral approaches to be most accurate 
when working with compression-based distance measures \cite{CO09,CO10,CO16}.
Mapping from clusters to classes for the pairwise analysis is done 
following the spectral clustering step by using a majority vote. 

\subsection{Example Applications}
We now describe results from a number of  sample applications. For all of these applications, we use a single implementation based on co-occurrence counts. For each search engine that we used, including Amazon, Wikipedia and NCBI a custom MATLAB script  was developed to parse the search count results. We used the page counts returned using the built in search from each website for the frequencies, and following the approach in \cite{CV07}  choose $N$ as the frequency for the search term 'the'.  The results described were not sensitive to the choice of search term used to establish $N$, for example identical classification results were obtained using the counts returned by the search term 'N' as the normalizing factor. Following each classification result below, we include in parenthesis the 95\% confidence interval for the result, computed as described in \cite{Witten2005}

The first three classification questions we considered used the wikipedia search engine. These questions include classifying colors vs. animals, classifying colors vs. shapes and classifying presidential candidates by political party for  the US 2008 U.S. presidential election. For colors vs animals and shapes, both pairwise and multiset NWD classified all of the elements 100\% correctly (0.82,1.0). For the presidential candidate classification by party, the pairwise NWD formulation performed poorly, classifying 58\% correctly (0.32,0.8), while the set formulation obtained 100\% correct classification (0.76,1.0). Table \ref{tab.WikipediaResults} shows the data used for each question, together with the pairwise and set accuracy and the total number of website queries required for each method. 

\begin{table}[tbp]
\centering
\includegraphics[width=8cm]{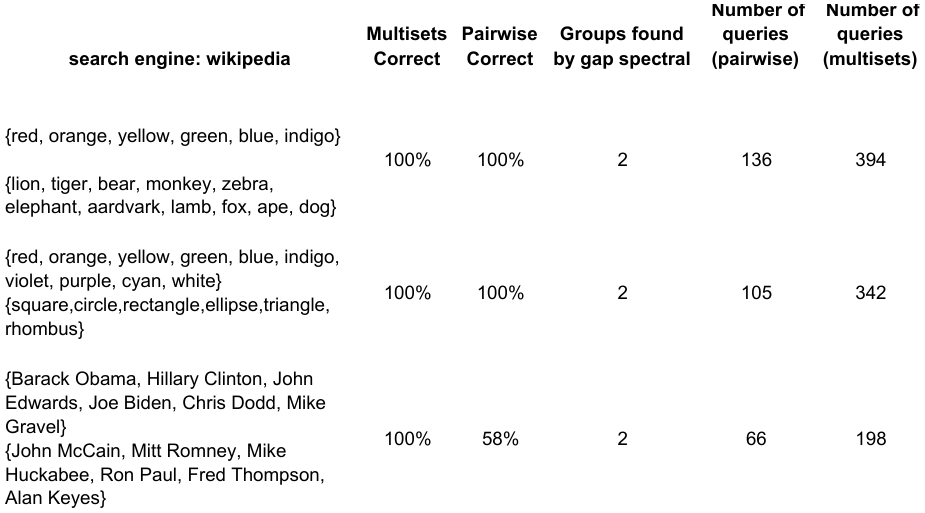}
\caption{Classification results using wikipedia.}
\label{tab.WikipediaResults}
\end{table}
The next classification question considered used page counts returned by the Amazon website search engine to classify book titles by author. Table \ref{tab.Amazon} summarizes the sets of novels associated with each author, and the classification results for each author as a confusion matrix. The Multiset NWD (top) misclassified one of the Tolstoy novels ('War and Peace') to Stephen King, but correctly classified all other novels correctly, 96\% accurate (0.83,0.99). The pairwise NWD performed significantly more poorly, achieving only 79\% accuracy (0.6,0.9).

\begin{table}[tbp]
\centering
\includegraphics[width=8cm]{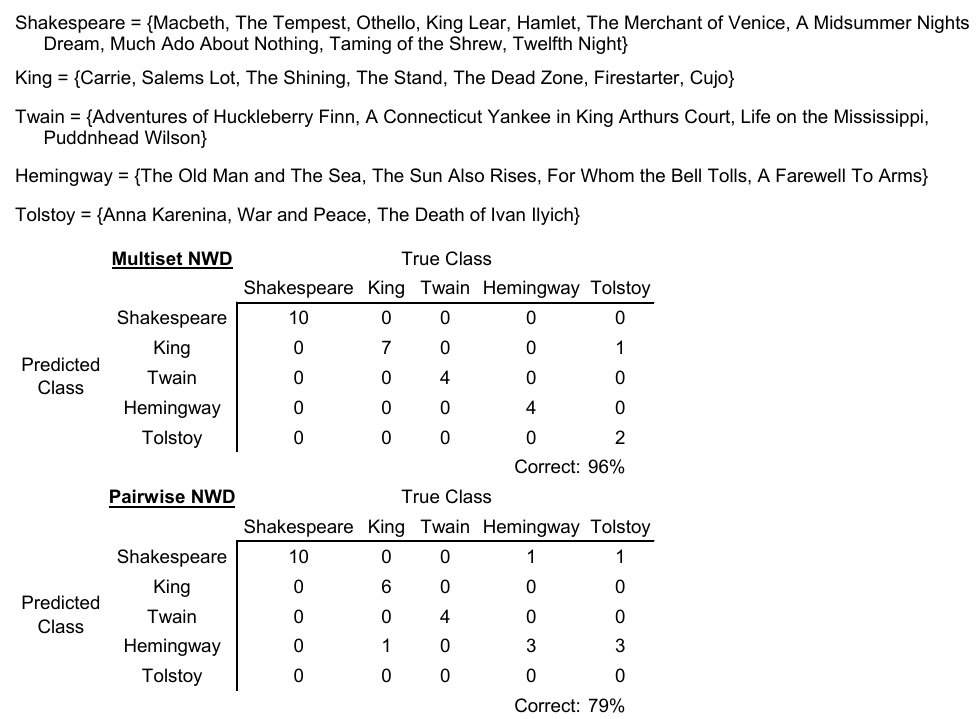}
\caption{Classifying novels by author using Amazon}
\label{tab.Amazon}
\end{table}
The final application considered is to quantify similarities among diseases based on the results of genome wide association studies (GWAS).  These studies scan the genomes from a large population of individuals to identify genetic variations occurring at fixed locations, or loci that can be associated with the given disease. Here we use the the NIH NCBI database to search for similarities among diseases, comparing loci identified by recent GWAS results for each disease. The diseases included Alzheimers \cite{Alzh}, Parkinsons \cite{Park}, Amyotrophic lateral sclerosis (ALS) \cite{ALS}, Schizophrenia \cite{Schiz}, Leukemia \cite{Leuk}, Obesity \cite{Obes}, and Neuroblastoma \cite{Neur}. The top of Table \ref{tab.pubmed} lists the loci used for each disease. The middle panel of Table \ref{tab.pubmed} shows at each location $(i,j)$ of the distance matrix the NWD computed for the combined counts for the loci of disease $i$ concatenated with disease $j$. The diagonal elements $(i,i)$ show the NWD for the loci of disease $i$. The bottom panel of Table \ref{tab.pubmed} shows the NWD for each element with the diagonal subtracted,  $(i,j)-(i,i)$. This is equivalent to the $NWD(Ax)-NWD(A)$ value used in the previous classification problems.  The two minimum values in the bottom panel, showing the relationships between Parkinsons and Obesity, as well as between Schizophrenia and Leukemia were surprising. The hypothesis was that neurological disorders such as Parkinsons, ALS and Alzheimers, would be more similar to each other. After these findings we found that there actually have been recent findings of strong relationships between both Schizophrenia and Leukemia \cite{SchizLeuk} as well as between Parkinsons and Obesity \cite{ParkObes}, relationships that have also been identified by clinical evidence not relating to GWAS approaches.

\begin{table}[tbp]
\centering
 \includegraphics[width=8cm]{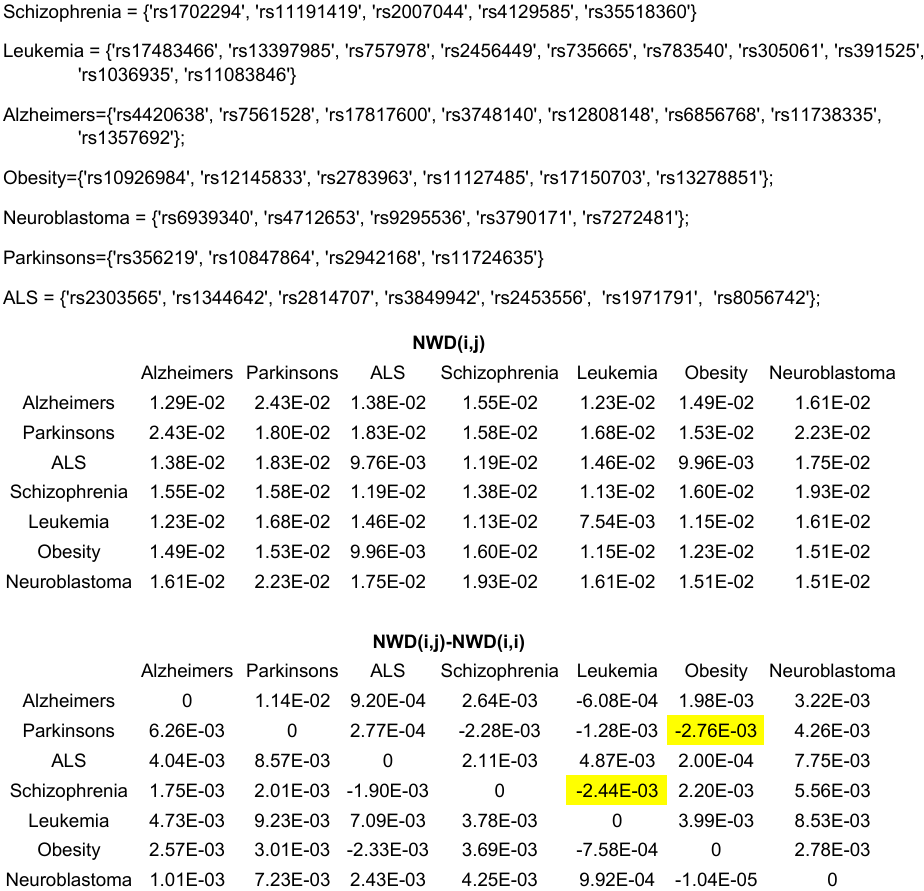}
\caption{GWAS loci from NIH NCBI input to NWD quantifies disease similarity.}
\label{tab.pubmed}
\end{table}


\section{Software Availability}
Free and open source (BSD) software implementations for the NWD are available from https://git-bioimage.coe.drexel.edu/opensource/nwd.

\section{Conclusion}\label{sect.concl}

Consider queries to a search engine using a data base divided
in chunks called web pages. On each query the search engine 
returns a set of web pages. Let $n$ be the cardinality of a query set
and $N$ the number of web pages in the data base multiplied by
the average number of search terms per web page.
We propose a method, the normalized web distance (NWD) for sets of
queries that quantifies in a single number between 0 and $(\log_n (N/n))/(n-1)$ 
the way in which the queries in the set are similar: 0 means 
all queries in the set are the same (the set has cardinality one) 
and $(\log_n (N/n))(n-1)$ means all queries in the
set are maximally dissimilar to each other.
The similarity among queries
uses the frequency counts of web pages returned for each query and the
set of queries.
The method can be applied using any big data
base and a search engine
that returns reliable aggregate page counts. Since this method uses names for 
the objects, and not the objects themselves, we can view 
the common similarity of the names
as a common semantics between those names (words or phrases).
The common similarity between a finite nonempty set of
queries can be viewed as a distance or diameter of this set. 
We show that this distance ranges in between 0 and 
$(\log_n (N/n))/(n-1)$, how it
changes under adding members to the set, that it does not
satisfy the triangle property, and that the NWD formally and provably 
expresses common similarity  (common semantics). 

To test the efficacy of the new method for classification 
we experimented with small data sets of queries
based on search results from Wikipedia, Amazon, and the 
National Center for Biotechnology Information (NCBI) website 
from the U.S. National Institutes of Health.
In particular we compared classification using pairwise NWDs (the NGDs) with
classification using set NWD. The last mentioned performed
consistently equal or better, sometimes much better.

\appendices
\section{Strings and the Self-Delimiting Property}\label{sect.set}
We write {\em string} to mean a finite binary string,
and $\epsilon$ denotes the empty string.
(If the string is over
a larger finite alphabet we recode it into binary.)
The {\em length} of a string $x$ (the number of bits in it)
is denoted by $|x|$. Thus,
$|\epsilon| = 0$. 
The {\em self-delimiting code} for $x$ of length $n$ is
$\bar{x} = 1^{|x|}0x$ of length $2n+1$, or even shorter
$x'=1^{\bar{x}}0x$ of length $n+2 \log n +1$ (see \cite{LV08} for still shorter
self-delimiting codes). Self-delimiting code words encode where they end.
The advantage is that 
if many strings of varying lengths are encoded self-delimitingly using the
same code, then their concatenation can be parsed in their constituent 
code words in one pass going from left to right. Self delimiting codes
are computable prefix codes. A {\em prefix code} has the property
that no code word is a proper prefix of any other code word. The code-word
set is called {\em prefix-free}.

We identify strings with natural numbers
by associating each string with its index
in the length-increasing lexicographic ordering according to the scheme
$
( \epsilon , 0),  (0,1),  (1,2), (00,3), (01,4), (10,5), (11,6), 
\ldots . 
$
In this way the Kolmogorov complexity can be about finite binary
strings or natural numbers.

\section{Computability Notions}\label{sect.comp}

A pair of integers
such as $(p,q)$ can be interpreted as the rational $p/q$.
We assume the notion of a  function with rational arguments
and values.
A function $f(x)$ with $x$ rational is \emph{upper semicomputable}
if it is defined by a rational-valued total computable function $\phi(x,k)$
 with $x$ a rational number
and $k$ a nonnegative integer
such that $\phi(x,k+1) \leq \phi(x,k)$ for every $k$ and
  $\lim_{k \rightarrow \infty} \phi (x,k)=f(x)$.
This means that $f$ can be computed from above
 (see \cite{LV08}, p. 35).  A function $f$ is  \emph{lower semicomputable}
if $-f$ is semicomputable from above.
 If a function is both upper semicomputable
and lower semicomputable then it is \emph{computable}.

\section{Kolmogorov Complexity}\label{sect.kolmcomp}
The Kolmogorov complexity is the information in a single finite object
\cite{Ko65}.
Informally, the Kolmogorov complexity of a finite binary string
is the length of the shortest string from which the original
can be lossless reconstructed by an effective
general-purpose computer such as a particular universal Turing machine.
Hence it constitutes a lower bound on how far a
lossless compression program can compress.
For technical reasons we choose Turing machines with a separate
read-only input tape that is scanned from left to right without backing up,
a separate work tape on which the computation takes place,
an auxiliary tape inscribed with the {\em auxiliary} information,
and a separate output tape. All tapes are divided into squares
and are semi-infinite. Initially,
the input tape contains a semi-infinite binary 
string with one bit per square
starting at the leftmost square, and all heads scan the leftmost squares
on their tapes. Upon halting, the initial segment $p$
of the input that has been scanned is called the input program
and the contents of the output tape is called the output.
By construction, the set of halting programs is prefix free 
(Appendix~\ref{sect.set}), and this type of Turing machine is called
a {\em prefix Turing machine}.
A standard enumeration of prefix Turing machines $T_1,T_2, \ldots$
contains a universal machine $U$ such that $U(i,p,y)=T_i(p,y)$ for all
indexes $i$, programs $p$, and auxiliary strings $y$. 
(Such universal machines are called
``optimal'' in contrast with universal machines like $U'$ 
with $U'(i,pp,y)=T_i(p,y)$
for all $i,p,y$, and $U'(i,q,y)=1$ for $q \neq pp$ for some $p$.)
We call $U$ the {\em reference universal prefix Turing machine}.
This leads to the definition of prefix Kolmogorov complexity.

Formally, the {\em conditional prefix Kolmogorov complexity}
$K(x|y)$ is the length of the shortest input $z$
such that the reference universal prefix Turing machine 
$U$ on input $z$ with
auxiliary information $y$ outputs $x$. The
{\em unconditional Kolmogorov complexity} $K(x)$ is defined by
$K(x|\epsilon)$ where $\epsilon$ is the empty string.
In these definitions both $x$ and $y$ can consist of strings into which
finite sets of finite binary strings are encoded.
Theory and applications are given in the textbook \cite{LV08}.

For a finite set of strings we assume that the strings
are length-increasing lexicographic ordered.
This allows us to assign
a unique Kolmogorov complexity to a set.
The conditional prefix Kolmogorov complexity $K(X|x)$ of a set
$X$ given an element $x$ is the length of a shortest program $p$
for the reference universal Turing machine that with input $x$
outputs the set $X$. The prefix Kolmogorov complexity $K(X)$ of a set
$X$ is defined by $K(X| \epsilon )$.
One can also put set in the conditional such as $K(x|X)$ or
$K(X|Y)$.
We will use the straightforward laws $K(\cdot|X,x)=K(\cdot|X)$
and $K(X|x)=K(X'|x)$ up to an additive constant term, for $x \in X$
and $X'$ equals the set $X$ with the element $x$ deleted.


We use the following notions from the theory of Kolmogorov complexity.
The {\em symmetry of information} property \cite{Ga74} for strings $x,y$ is 
\begin{equation}\label{eq.soi}
K(x,y)=K(x)+K(y|x)=K(y)+K(x|y),
\end{equation}
with equalities up to an additive term $O(\log(K(x,y)))$.

\section{Metricity}\label{sect.metric}
A {\em distance function} $d$
on ${\cal X}$ is defined by $d:{\cal X} \rightarrow {\cal R}^+$ where
${\cal R}^+$ is the set of nonnegative real numbers.
If $X,Y,Z \in {\cal X}$, then $Z=XY$ if
$Z$ is the set consisting of the elements of the sets $X$ and $Y$
ordered length-increasing lexicographic.
A distance function
$d$ is a {\em metric} if
\begin{enumerate}
\item
{\em Positive definiteness}: $d(X)=0$ if all elements of $X$ are equal
and $d(X) > 0$ otherwise. (For sets equality of all members means $|X|=1$.)
\item
{\em Symmetry}: $d(X)$ is invariant
under all permutations of $X$.
\item
{\em Triangle inequality}: $d(XY) \leq d(XZ)+d(ZY)$.
\end{enumerate}

\section{Proofs}\label{sect.proofs}

\begin{proof} of Lemma~\ref{lem.egmax}.

Run all programs dovetailed fashion and at each 
time instant select a shortest
program that with inputs $e(x)$ for all $x \in X$ has terminated with 
the same output $e(X)$. The lengths of these shortest programs
gets shorter and shorter, 
and in for growing time eventually reaches $EG_{\max}(X)$ (but we do not know
the time for which it does).
Therefore $EG_{\max}(X)$ is upper semicomputable.  
It is not computable since for $X=\{x,y\}$ we have $EG_{\max}(X)=
\max\{K(e(x)|e(y)),K(e(y)|e(x))\}+O(1)$, the information distance 
between $e(x)$ and $e(y)$ which is known to be incomputable \cite{BGLVZ98}.
\end{proof}

\begin{proof} of Theorem~\ref{theo.just}.

($\leq$) 
We use a modification of the proof of \cite[Theorem 2]{Li08}.
According to Definition~\ref{def.webevent} $x=y$ iff $e(x)=e(y)$.
Let $X= \{x_1, \ldots , x_n\}$ and $k=\max_{x \in X} \{K(e(X)|e(x)\}$. 
A set of cardinality $n$ in $S$ is for the purposes of this proof
represented by an $n$-vector of which the
entries consist of the lexicographic length-increasing sorted members 
of the set. For each $1 \leq i \leq n$ let ${\cal Y}_i$ be the set of 
computably enumerated $n$-vectors $Y=(y_1, \ldots, y_n)$ with entries in $S$ 
such that $K(e(Y)|e(y_i)) \leq k$ for each $1 \leq i \leq n$.
Define the set $V = \bigcup_{i=1}^n {\cal Y}_i$. 
This $V$ is the set of vertices
of a graph $G=(V,E)$. The set of edges $E$ is defined by: 
two vertices $u=(u_1, \ldots , u_n)$
and $v=(v_1, \ldots ,v_n)$ are connected by an edge 
iff there is $1 \leq j \leq n$ such that $u_j=v_j$. 
There are at most $2^k$
self-delimiting programs of length at most $k$ computing from input 
$e(u_j)$ to different $e(v)$'s with $u_j$ in vertex $v$ as $j$th entry. 
Hence there
can be at most $2^k$ vertices $v$ with $u_j$ as $j$th entry. Therefore,
for every $u \in V$
and $1 \leq j \leq n$ there are at most $2^k$ vertices
$v \in V$ such that $v_j=u_j$. 
The vertex-degree of graph $G$ is
therefore bounded by $n2^k$. Each graph can be vertex-colored 
by a number of colors equal to the maximal vertex-degree. This divides
the set of vertices $V$ into disjoint color classes
$V=V_1 \bigcup \cdots \bigcup V_D$ with $D \leq n2^k$. To compute $e(X)$
from $e(x)$ with $x \in X$ we only need the color class of which $e(X)$
is a member and the position of $x$ in $n$-vector $X$.  
Namely, by construction every vertex with the same
element in the $j$th position is connected by an edge. Therefore there is
at most a single vertex with $x$ in the $j$th position 
in a color class.
Let $x$ be the $j$th entry of $n$-vector $X$.
It suffices to have a program of length at most 
$\log (n2^k) +O(\log nk) = k+ O(\log nk)$ bits to compute $e(X)$ from $e(x)$. 
From $n$ and $k$ we can generate $G$ and given $\log (n2^k)$ bits we
can identify the color class $V_d$ of $e(X)$. Using another $\log n$ bits
we define the position of $x$ in the $n$-vector $X$.
To make such a program self-delimiting add a logarithmic term. In total 
$k+O(\log k)$ suffices since $O( \log k)=O(\log n + \log nk)$.

($\geq$) That $EG_{\max}(X) \geq \max_{x \in X} \{K(e(X)|e(x)\}$ follows
trivially from the definitions.
\end{proof}

\begin{proof} of Lemma~\ref{lem.0}.

($\geq 0$) 
Since $f(X) \leq f(x)$ for all $x \in X$ the numerator of
the right-hand side of \eqref{eq.NWD} is nonnegative. 
Since the denominator is also nonnegative we have $NWD(X) \geq 0$.
Example of the lower bound: if 
$\max_{x \in X}\{\log f(x)\}= \log f(X)$, then $NWD(X)=0$.

($\leq (\log_{|X|} (N/|X|)) /(|X|-1)$) 
Write $n=|X|$,  $x_M= \arg\max_{x \in X} f(x)$ and 
$x_m= \arg\min_{x \in X} f(x)$.
Rewrite \eqref{eq.NWD} as $(n-1)NWD(X)= \log (f(x_M)/f(X))/\log (N/f(x_m))$.  
This expression can only reach its maximum if 
$f(X)$ is as small as possible which can be achieved independent 
of the other parameters. 
To this end the web events $e(x)$ for $x \in X$
satisfy $\bigcap_{x \in X} e(x)$ is a singleton set which means that $f(X)=1$.
(For $f(X)=0$ we have $\bigcap_{x \in X} e(x)= \emptyset$ and $NWD(X)$ is
undefined.) 
For $f(X)=1$ the expression can be rewritten as 
$(n-1)NWD(X)= \log_{N/f(x_m)} f(x_M) = \alpha$
where $\alpha$ is determined by $(N/f(x_m))^{\alpha}=f(x_M)$.
The side conditions which must be satisfied are $f(x_m) \leq f(x_M)$ and 
$(n-1)f(x_m)+f(x_M) \leq N$. 
For any fixed $f(x_M)$ the value of $\alpha$ is maximal
if $f(x_m)$ is as large
as possible which means that $f(x_m)= f(x_M)$.
Then $f(x_M)= N^{\alpha/(\alpha+1)}$. 
With $\bigcup_{x \in X}e(x)=\Omega$
and $\bigcap_{x \in X}e(x)$ is a singleton set we have 
$f(x_M)=(N-1)/n +1$. It follows that
$\log ((N+n-1)/n) = (\alpha/(\alpha+1)) \log N$. Rewriting yields first
$1-\log_N ((N+n-1)/n) = 1/(\alpha +1)$ and then 
$\alpha = (1/(1-\log_N ((N+n-1)/n)))-1= (1/\log_N(Nn/(N+n-1)))-1$.
Hence $NWD(X) \leq (1/\log_N(Nn/(N+n-1))-1)/(n-1) < (1/\log_N n -1)/(n-1)
= (\log_n (N/n))/(n-1)$.
\end{proof}

\begin{proof} of Lemma~\ref{claim.1}.

{\rm (i)} Since $X \subseteq Y$ and because of 
the condition of item (i) we have 
$\min_{y \in Y}\{\log f(y)\} = \min_{x \in X}\{\log f(x)\}$. From 
$X \subseteq Y$ also follows 
$\max_{y \in Y}\{\log f(y)\} \geq \max_{x \in X}\{\log f(x)\}$, and
$\log f(X) \geq \log f(Y)$.
Therefore the numerator
of $NWD(Y)$ is at least as great as that of $NWD(X)$, and
the denominator
of $NWD(Y)$ equals $(|Y|-1)/(|X|-1)$ times the denominator of $NWD(X)$.

{\rm (ii)} We have
$\min_{x \in Y} \log f(y) < \min_{x \in X} \{\log f(x)\}$.
If $NWD(X)$ is maximal
then $NWD(Y)$ is maximal (in both cases there is least common similarity
of the members of the set). Item {\rm (ii)} follows vacuously in this case.
Therefore assume that $NWD(X)$ is less than maximal. 
Write $NWD(X)=a/b$ with $a$ equal to the numerator of $NWD(X)$ and $b$
equal to the denominator. If $c,d$ are real numbers satisfying
$c/d \geq a/b$ then $bc \geq ad$. Therefore $ab+bc \geq ab+ad$ which
rearranged yields $(a+c)/(b+d) \geq a/b$. If $c/d < a/b$ then by similar
reasoning $(a+c)/(b+d) < a/b$.

Assume \eqref{eq.cond} holds.
We take the logarithms of both sides of \eqref{eq.cond} and rearrange it
to obtain $\log f(X) - \max_{x \in X}\{\log f(x)\}- 
\log f(Y) + \max_{y \in Y} \{\log f(y)\} \geq 
(\min_{x \in X}\{\log f(x)\}- \min_{y \in Y} \{\log f(y)\})(|X|-1)NWD(X)$.
Let the lefthand side of
the inequality be $c$ and the righthand side of the inequality be $dNWD(X)$. 
Then
\begin{align}\label{eq.nwdxy}
NWD(X) & = \frac{\max_{x \in X} \{\log f(x)\}- \log f(X)}
{(\log N - \min_{x \in X} \{\log f(x)\})(|X|-1)} 
\\& \leq \frac{\max_{y \in Y} \{\log f(y)\}- \log f(Y)}
{(\log N - \min_{y \in Y} \{\log f(y)\})(|X|-1)}
\nonumber
\\&=\frac{|Y|-1}{|X|-1}NWD(Y).
\nonumber
\end{align}
The inequality holds by the rewritten \eqref{eq.cond}
and the $a,b,c,d$ argument above since $c/d \geq NWD(X)=a/b$. 

Assume \eqref{eq.cond} does not hold, that is, it holds with 
the $\geq$ sign replaced by a $<$ sign.
We take logarithms of both sides of this last version and rewrite it
to obtain $\log f(X) - \max_{x \in X}\{\log f(x)\}- 
\log f(Y) + \max_{y \in Y} \{\log f(y)\} <
(\min_{x \in X}\{\log f(x)\}- \min_{y \in Y} \{\log f(y)\})(|X|-1)NWD(X)$.
Let the lefthand side of the inequality be $c$ and the righthand 
side $dNWD(X)$. Since $c/d <NWD(X)=a/b$ we have $a/b >(a+c)/(b+d)$ by
the $a,b,c.d$ argument above. 
Hence \eqref{eq.nwdxy} holds with 
the $\leq$ sign switched to a $>$ sign.
It remains to prove that $NWD(Y) \geq NWD(Z)(|Z|-1)/(|Y|-1)$.
This follows directly from item (i).
\end{proof}

\begin{proof} of Lemma~\ref{theo.triangle}.

The following is a counterexample.
Let $X=\{x_1\}$, $Y=\{x_2\}$, $Z=\{x_3,x_4\}$, 
$\max_{x \in XY} \{\log f(x)\} =10$, $\max_{x \in XZ} \{\log f(x)\}=10$, 
$\max_{x \in ZY} \{\log f(x)\}=5$,
$\log f(XY)= \log f(XZ)= \log f(ZY)=3$, 
$\min_{x \in XY}\{\log f(x)\} = \min_{x \in XZ}\{f(x)\}
=\min_{x \in ZY}\{\log f(x)\} =4$, and $\log N = 35$. This arrangement
can be realized for queries $x_1,x_2,x_3,x_4$. 
(As usual we assume 
that $e(x_i) \neq e(x_j)$
for $1 \leq i,j \leq 4$ and $i \neq j$.)
Computation shows
$NWD(XY)>NWD(XZ)+NWD(ZY)$ since $7/31 > 7/62+1/62$.
\end{proof}

\begin{proof} of Theorem~\ref{theo.ideal}.

We start with the following:
\begin{claim}\label{claim.wdf}
\rm
$EG_{\max}(X)$ is an admissible web distance function and
$EG_{\max}(X) \leq D(X)$ for every computable 
admissible web distance function $D$.
\end{claim}
\begin{proof}
Clearly $EG_{\max}(X)$ satisfies items (i) and (ii) of 
Definition~\ref{def.wdf}. To show it is an admissible web 
distance it remains to
establish the density requirement (iii). For fixed $x$ consider
the sets $X \ni x$ and $|X| \geq 2$. We have
\[
\sum_{X:  X \ni x \;\&\;|X| \geq 2}2^{-EG_{\max}(X)}
\leq 1,
\]
since for every $x$ the set $\{EG_{\max}(X): X \ni x\;\&\;EG_{\max}(X)>0\}$ 
is the length set of a binary prefix code and 
therefore the summation above satisfies the Kraft inequality 
\cite{Kr49} given by \eqref{eq.kraft}. 
Hence $EG_{\max}$ is an admissible distance.

It remains to prove minorization.
Let $D$ be a computable admissible web distance, 
and the function $f$ defined 
by $f(X,x) = 2^{-D(X)}$ for $x \in X$ and 0 otherwise.
Since $D$ is computable the function $f$ is computable. 
Given $D$, one can
compute $f$ and therefore $K(f) \leq K(D)+O(1)$. 
Let ${\bf m}$ denote
the universal distribution \cite{LV08}. 
By \cite[Theorem 4.3.2]{LV08} $c_D {\bf m}(X|x) \geq f(X,x)$ with 
$c_D=2^{K(f)}=2^{K(D)+O(1)}$, that is, $c_D$ is
a positive constant depending on $D$ only. By \cite[Theorem 4.3.4]{LV08}
we have $-\log {\bf m}(X|x)=K(X|x)+O(1)$. Altogether,
for every $X \in {\cal X}$ and for every $x \in X$ holds
$\log 1/ f(X,x) \geq K(X|x)+\log 1/c_D +O(1)$. 
Hence $D(X) \geq EG_{\max}(X) +\log 1/c_D+O(1)$.
\end{proof}

By Lemma~\ref{lem.egmax} the function $EG_{\max}$ is upper semicomputable
but not computable. The function $G(X)-\min_{x \in X}\{G(x)\}$ is a computable 
and an admissible function as in Definition~\ref{def.wdf}.
By Claim~\ref{claim.wdf} it is an upper bound on $EG_{\max}(X)$ and hence 
$EG_{\max}(X) < G(X)-\min_{x \in X}\{G(x)\}$.
Every admissible property or feature that is common to all members
of $X$ is quantized as an upper bound on $EG_{\max}(X)$. Thus, 
the closer $G(X)-\min_{x \in X}\{G(x)\}$ approximates
$EG_{\max}(X)$, the better it approximates
the common admissible properties among all search terms in $X$. This 
$G(X)-\min_{x \in X}\{G(x)\}$ is
the numerator of $NWD(X)$. The denominator is 
$\max_{x \in X} \{G(x)\}(|X|-1)$, a
normalizing factor.
\end{proof}

\section*{Funding} \addcontentsline{toc}{section}{Acknowledgment}
Portions of this research were supported by  the National Institute On Aging 
of the National Institutes of Health under award number R01AG041861 to A. R. Cohen.

\section*{Conflict of Interest} \addcontentsline{toc}{section}{Conflict of Interest}
The authors declare that they have no conflict of interest.

\bibliographystyle{natbib}

\begin{thebibliography}{99}

\bibitem{ALS}
A.K. Ahmeti et al. Age of onset of amyotrophic lateral 
sclerosis is modulated by a 
locus on 1p34.1, {\em Neurobiology of Aging} 34:1(2013), 357.e357-357.e319.

\bibitem{BbA05}
J.P. Bagrow and D. ben-Avraham, On the Google-fame of scientists and other
populations, {\em AIP Conference Proceedings}
779:1(2005), 81--89.

\bibitem{BGLVZ98}
C.H. Bennett, P. G\'acs, M. Li, P.M.B. Vit\'anyi, and W. Zurek,
Information distance, {\em IEEE Trans. Inform. Theory},

44:4(1998), 1407--1423.



\bibitem{CV07}
R.L. Cilibrasi and P.M.B. Vit\'anyi, The Google similarity distance, 
{\em IEEE Trans. Knowledge and Data Engineering}, 19:3(2007), 370-383.

\bibitem{CS04}
P. Cimiano and S. Staab, Learning by Googling,
{\em SIGKDD Explorations}, 6:2(2004), 24--33.

\bibitem{ParkObes}
H. Chen, et al., Obesity and the risk of Parkinson's disease,
{\em Am. J. Epidemiol.}, 159:6(2004), 547--555.

\bibitem{CO09}
A.R. Cohen, C. Bjornsson, S. Temple, G. Banker and B. Roysam, Automatic Summarization of Changes in Biological Image Sequences using Algorithmic Information Theory, {\em IEEE Trans. Pattern Anal. Mach. Intell.} 31(8):(2009) 1386-1403.

\bibitem{CO10}
A.R. Cohen, F. Gomes,  B.Roysam, and M. Cayouette, 
Computational prediction of neural progenitor cell fates, 
{\em Nature Methods}, 7:3(2010), 213--218.


\bibitem{CV13}
A.R. Cohen and P.M.B. Vit\'anyi, Normalized compression distance of multisets 
with applications, {\em IEEE Trans. Pattern Analysis Machine Intelligence},
37:8(2015), 1602--1614.

\bibitem{Ga74}
P. G\'acs,
On the symmetry of algorithmic information,
{\em Soviet Math. Doklady}, 15:1477--1480, 1974.
Correction, Ibid., 15(1974), 1480.


\bibitem{SchizLeuk}
H.S. Huang, et al., Prefrontal dysfunction in schizophrenia involves 
mixed-lineage leukemia 1-regulated histone methylation at GABAergic gene 
 promoters, {\em J. Neuroscience} 27:42(2007), 11254--11262.

\bibitem{CO16}
R.Joshi, et al., Automated measurement of cobblestone morphology for 
characterizing stem cell derived retinal pigment epithelial cell cultures, 
{\em J. Ocular Pharmacology Therapeutics}, 32:5(2016),331--339.

\bibitem{Alzh}
M.I. Kamboh,  et al. Genome-wide association study of Alzheimer's disease ,{\em Translational Psychiatry - Nature} 2 (2012): e117.

\bibitem{KL05}
F. Keller and M. Lapata,
Using the web to obtain frequencies for unseen bigrams,
{\em Computational Linguistics}, 29:3(2003), 459--484.

\bibitem{Ko65}
A.N. Kolmogorov,
{Three approaches to the quantitative definition of information},
{\em Problems Inform. Transmission} 1:1(1965), 1--7.

\bibitem{Kr49}
L.G. Kraft, A device for quantizing, grouping, and coding amplitude 
modulated pulses, MS Thesis, EE Dept., Massachusetts Institute of Technology,
Cambridge. Mass., USA, 1949.

\bibitem{LD97}
T. Landauer and S. Dumais, A solution to Plato's problem:
The latent semantic analysis theory of acquisition,
induction and representation of knowledge,
{\em Psychol. Rev.}, 104(1997), 211--240.

\bibitem{LCB}
Y. LeCun, C. Cortes and C.J.C. Burges, The MNIST database of handwritten digits,
http://yann.lecun.com/exdb/mnist/

\bibitem{Le74}
L.A. Levin,
Laws of information conservation (nongrowth) 
and aspects of the foundation of probability theory,
{\em Probl. Inform. Transm.}, 10(1974), 206--210.

\bibitem{LV08}
M. Li and P.M.B. Vit\'anyi.
{\em An Introduction to Kolmogorov Complexity
and its Applications}, Springer-Verlag, New York, Third edition, 2008.

\bibitem{Li08}
C. Long, X. Zhu, M. Li and B. Ma, Information shared by many objects,
Proc. 17th ACM Conf. Information and Knowledge Management,
2008, 1213--1220.

\bibitem{Neur}
J.M. Maris, Chromosome 6p22 Locus Associated with Clinically Aggressive 
Neuroblastoma, {\em New England Journal of Medicine} 358:24(2008), 2585--2593.


\bibitem{Mc56}
B. McMillan, Two inequalities implied by unique decipherability, 
{\em IEEE Trans. Information Theory}, 2:4(1956), 115–-116.

\bibitem{Michel11}
J.-B. Michel, Y.K. Shen, A.P. Aiden, A. Veres, M.K. Gray, T.G.B. Team, et al., 
Quantitative Analysis of Culture Using Millions of Digitized Books, 
{\em Science}, 331(2011), 176--182, (January 14 2011).

\bibitem{MCCD13}
T. Mikolov, K. Chen, G. Corrado and J. Dean,
Efficient estimation of word representations in vector space,
{\em ICLR Workshop}, 2013. Also arXiv:1301.3781.
 
\bibitem{Ng02}
A.Y. Ng, M. Jordan and Y. Weiss, On Spectral Clustering: Analysis 
and an algorithm, {\em Advances Neural Informat. Process. Systems}, 14, (2002).

\bibitem{Obes}
A. Scherag, et al., Two New Loci for Body-Weight Regulation Identified in a Joint Analysis of Genome-Wide Association Studies for Early-Onset Extreme Obesity in French and German Study Groups, {\em PLoS Genetics}, 6:4(2010), e1000916.

\bibitem{Schiz}
Schizophrenia Working Group of the Psychiatric Genomics Consortium,
Biological insights from 108 schizophrenia-associated genetic loci, 
{\em Nature} 511(7510), 2014, 421-427.

\bibitem{Sh48}
C.E. Shannon, The mathematical theory of communication,
{\em Bell System Tech. J.}, 27(1948), 379--423, 623--656.

\bibitem{Leuk}
F.C.M. Sill\'e, et al., Post-GWAS Functional Characterization 
of Susceptibility Variants for Chronic Lymphocytic Leukemia, {\em PLoS One},
7:1(2012), e29632.

\bibitem{Park}	
A.I. Soto-Ortolaza, A. I. et al., GWAS risk factors in Parkinson's disease: LRRK2 coding variation and genetic interaction with PARK16, {\em Am. J. 
Neurodegener Dis.} 2:4(2013), 287--299.


\bibitem{TKS02}
P.-N. Tan, V. Kumar and J. Srivastava, Selecting the right
interestingness measure for associating patterns. {\em Proc.
ACM-SIGKDD Conf. Knowledge Discovery and Data Mining},
2002, 491--502.


\bibitem{TC03}
E. Terra and C.L.A. Clarke, Frequency estimates for
statistical word similarity measures, 37/162 in
Human Language Theory Conference
(HLT/NAACL 2003), Edmonton, Alberta, 2003. 



\bibitem{Vi11}
P.M.B. Vit\'anyi, Information distance in multiples, 
{\em IEEE Trans. Inform. Theory}, 57:4(2011), 2451-2456. 

\bibitem{Witten2005}
I.H. Witten and E. Frank, {\em Data Mining: Practical Machine Learning 
Tools and Techniques}, 2005.


\end{thebibliography}

\end{document}